\documentclass[12pt,twoside]{article}
\usepackage{amssymb}
\usepackage{amsmath}
\usepackage[makeroom]{cancel}
\usepackage{latexsym}
\usepackage{longtable}
\usepackage{epsfig}
\usepackage{enumitem}
\usepackage{graphicx,bbm,psfrag}
\usepackage{xcolor}

\usepackage[unicode=true,pdfusetitle,
 bookmarks=true,bookmarksnumbered=false,bookmarksopen=false,
 breaklinks=true,pdfborder={0 0 0},backref=false,colorlinks=true]
{hyperref}
\hypersetup{
 linkcolor=blue,citecolor=blue, urlcolor=blue}

\graphicspath{{images/}}
\newcommand{\bc}{\color{blue}}
\newcommand{\rc}{\color{black}}

\def\qed{$\Box$\medskip}

\newcommand{\mcFn}{\mathfrak{H}}

\renewcommand{\nu}{d}

\newcommand{\beqa}{\begin{eqnarray}}
\newcommand{\eeqa}{\end{eqnarray}}

\newcommand{\mrm}{\mathrm}

\renewcommand{\mathbf}{\boldsymbol}

\renewcommand{\i}{\mathrm i}

\newcommand{\ti}{\tilde}

\newcommand{\Om}{\Omega}
\newcommand{\ga}{\gamma}
\newcommand{\Si}{\Sigma}

\newcommand{\ka}{\kappa}

\newcommand{\be}{\beta}

\newcommand{\pa}{\partial}

\newcommand{\ov}{\overline}

\newcommand{\eps}{\epsilon}
\newcommand{\de}{\delta}
\newcommand{\De}{\Delta}
\newcommand{\e}{\mathrm{e}}

\newcommand{\hP}{\hat{P}}

\newcommand{\si}{\sigma}
\newcommand{\ph}{\phantom}
\newcommand{\h}{\fr{1}{2}}
\newcommand{\nat}{\mathbb{N}}
\newcommand{\hil}{\mathcal{H}}
\newcommand{\om}{\omega}

\newcommand{\fr}[2]{\frac{#1}{#2}}
\newcommand{\al}{\alpha}
\newcommand{\real}{\mathbb{R}}
\newcommand{\complex}{\mathbb{C}}
\newcommand{\la}{\lambda}
\newcommand{\non}{\nonumber}

\newcommand{\lan}{\langle}
\newcommand{\ran}{\rangle}

\newtheorem{theorem}{Theorem}[section]
\newenvironment{proof}{{\bf Proof}:}{\hspace*{\fill}$\Box$ \vspace{0.2cm}}

\newtheorem{definition}[theorem]{Definition}
\newtheorem{proposition}[theorem]{Proposition}
\newtheorem{lemma}[theorem]{Lemma}

\newtheorem{con}[theorem]{Conjecture}

\begin{document}
\vspace*{2cm}
\begin{center}
 \LARGE{Effective mass of the polaron - revisited }\bigskip
 \end{center}
\begin{center}
\large{Wojciech Dybalski, Herbert Spohn}
 \end{center}
 \begin{center}
Zentrum Mathematik and Physik Department, TUM,\\
Boltzmannstra{\ss}e 3, 85747 Garching, Germany.\\
 \tt{dybalski@ma.tum.de, spohn@tum.de}
 \end{center}

\vspace*{3.5cm}
\noindent
\textbf{Abstract}. Properties of the energy-momentum relation  for the Fr\"ohlich polaron
 are of continuing interest, especially for large values of the coupling constant.  
By combining spectral theory with the available results on the central limit theorem for the polaron path
measure we prove  that, except for an intermediate range of couplings,  the inverse effective mass is strictly positive and coincides with the diffusion constant.  
Such a result is established also for polaron-type models with a suitable ultraviolet cut-off and for arbitrary values of the coupling constant.  We point out  a slightly
stronger variant of the central limit theorem which would imply that the energy-momentum relation has a \emph{unique} global minimum
attained  at zero momentum. \vspace*{4cm}
\begin{flushright} 08.8.2019
\end{flushright}
\newpage
%%%%%%%%%%%%%%%%%%%%%%%%%%%%%%%%%%%
\section{Introduction}\label{sec1}
\setcounter{equation}{0} 
%%%%%%%%%%%%%%%%%%%%%%%%%%%%%%%%%%%

Polaron refers to an electron interacting with the lattice vibrations of a polar crystal, see \cite{AD10,GL91,S87} as a guide to the physics literature. In the conventional approximations the quantum  Hamiltonian  reads
\begin{equation}\label{1.1}
H = \tfrac{1}{2} p^2  + \int _{\mathbb{R}^d}\mathrm{d}k \,\omega(k) a^*(k) a(k) + \sqrt{\alpha} \int _{\mathbb{R}^d}\mathrm{d}k\frac{\hat{v}(k)}{\sqrt{2\omega(k)}}\big(\mathrm{e}^{\mathrm{i}kx}
a(k) + \mathrm{e}^{-\mathrm{i}kx}a^*(k)\big).
\end{equation}
We use units in which the bare electron mass equals one. $x,p$ are position and momentum of the electron in $\mathbb{R}^d$, $a^*(k),a(k)$ are the creation and annihilation operators of a free scalar Bose field over $\mathbb{R}^d$ with commutation relations $[a(k) ,a^*(k')] = \delta(k-k')$, $\omega$ is the dispersion relation of the Bose field, $\omega \geq 0$, continuous, and $\omega(Rk) = \omega(k)$
for all rotations $R$. The form factor $\hat{v}$ is assumed to be real rotation invariant and  has the  inverse Fourier transform\footnote{We use the convention $\hat v(k):=\fr{1}{(2\pi)^{d/2}}\int_{\real^d} dx\, \e^{-\mrm{i}kx}v(x)$ and occasionally write
$(Fv)(k):=\hat{v}(k)$.} $\hat{v}$.  $v(x)$ physically describes the smearing of the interaction between the electron and the Bose field. 
It is a standard convention to call  $g = \hat{v}/\sqrt{2\omega}$ the coupling function. Finally $\alpha$ is the coupling constant, $\alpha \geq 0$. Formally $H$ acts on the Hilbert space $\mathcal{H} = L^2(\mathbb{R}^d)\otimes \mathcal{F}$ with
$\mathcal{F}$ the Fock space of the Bose field. The coupling between field and particle is translation invariant and hence the total momentum
\begin{equation}\label{1.2}
P = p +P_\mathrm{f},\qquad P_\mathrm{f} = \int _{\mathbb{R}^d}\mathrm{d}k\,k \,  a^*(k)a(k),
\end{equation}
is conserved.

The \textit{Fr\"{o}hlich polaron} corresponds to the specific choice $d=3$, $\omega(k) = 1$, and $g(k) = ( \sqrt{2}\pi   |k|)^{-1}$. In particular $\|g\|_{2} = \infty$ because of ultraviolet divergence. Since the strong coupling physics is dominated by the large $|k|$ behavior of the coupling function, no ultraviolet cut-off can be afforded and  a separate discussion is required, see  Section~\ref{Froehlich-section}.
The acoustic polaron 
corresponds to $\omega(k) = |k|$ and other variations can be found in the physics literature.  As common practice \cite{S87,SM14}, we thus keep $d,\omega,g$ general for a while and add further assumptions on the way. 
For the purpose of this introductory discussion $\|g\|_{2} < \infty$ is assumed. Under precise conditions to  be stated in Section~\ref{polaron-type-section}, $H$ is a self-adjoint operator and has the fiber decomposition 
\begin{equation}\label{1.3}
H =\Pi^*\big(\int^{\oplus}_{\mathbb{R}^d} \mathrm{d}P\, H(P)\big)\Pi,
\end{equation}
since $P$ is conserved. (Here  $\Pi:=F \e^{\mrm{i} P_{\mrm{f}}x } $, where $F$ is the Fourier transform from $x$ to $P$ variable). The fiber Hamiltonian reads
\begin{equation}\label{1.4}
H(P) = \tfrac{1}{2}( P -P_\mathrm{f})^2 + \int _{\mathbb{R}^d}\mathrm{d}k \,\omega(k) a^*(k) a(k) + \sqrt{\alpha} \int _{\mathbb{R}^d}\mathrm{d}k \,g(k)\big(
a(k) + a^*(k)\big)
\end{equation}
and acts on $\mathcal{F}$.
The energy-momentum relation, $E(P)$, is the bottom of the spectrum of 
 $H(P)$,
\begin{equation}\label{1.5}
E(P) = \mathrm{inf}\, \mrm{spec} (H(P)).
\end{equation}
By construction $E(RP) = E(P)$ for all rotations $R$ and hence $E(P) = E_\mathrm{r}(|P|)$. Also, in generality,
\begin{equation}\label{1.6}
E(P) \geq E(0).
\end{equation}
As a widely accepted definition, the effective mass is the inverse of the curvature of $E(P)$ at $P=0$, which by rotation invariance means
\begin{equation}\label{1.7}
(m_\mathrm{eff})^{-1} = E_\mathrm{r}''(0).
\end{equation}

A long-standing open problem is to analyse the  effective mass of the Fr\"{o}hlich polaron  in the strong coupling regime 
\cite{P54,S87}.
As the coupling is increased, more and more bosons are bound to the electron and one would expect the effective mass to increase with 
$\alpha$, presumably to diverge in the limit. Recent progress has been achieved by Lieb and Seiringer \cite{LS19}, who  by functional 
analytic methods prove that indeed
$m_\mathrm{eff}(\alpha)$ diverges as $\alpha \to \infty$.
A mathematically orthogonal approach is to study the polaron path measure, as originally introduced by Feynman \cite{F55}, see also \cite{F,Ro94}. 
Mathematically this corresponds to a standard Brownian motion with a Gibbsian-like weight which depends only on the increments.
Thus one expects  to still observe diffusive behavior on large scales, however with an effective diffusion constant, $\sigma^2$.
In other words, one  conjectures the validity of a central limit theorem (CLT) for such weighted Brownian motion. 
In fact, the CLT has become now available for a large class of polaron models, including the assertion that $\sigma^2 >0$. We refer to Section~\ref{sec2} for more details.

As argued in   \cite{S87}, and presumably before,  effective mass and diffusion constant should be related as 
\begin{equation}\label{1.10}
(m_\mathrm{eff})^{-1} = \sigma^2.
\end{equation}
However,  at the time the reasoning was based on considering second moments for the position of the weighted Brownian motion,
a piece of information which is not so easily available from current CLT proofs.   For us this by itself is a convincingly enough reason to 
reconsider the case. As an extra bonus, apparently not noted before, the conventional CLT also yields spectral information about properties apparently unaccessible by current functional analytic techniques. 
To explain this point in more detail we assume the lower bound $\omega(k) \geq c_0 >0$ and further conditions as stated in Sections~\ref{polaron-type-section}  
and \ref{Froehlich-section}. Then $H(P)$ has a unique ground state with energy $E(P)$ for $|P|<\ga$ for some $\ga>0$.
% inside some ball $|P|<\ga$, $\ga>0$. 
%$\{ P\in \real^d\,|\,  |P| < \gamma\,\}$,  $\gamma > 0$. 
Furthermore the continuum edge of $H(P)$ is strictly larger than $E(P)$  in this ball. Possible eigenvalues have a finite multiplicity and can accumulate only at the continuum edge \cite{F74,SM14}. In particular, $E_\mathrm{r}$ is real analytic in $P$ and $\sqrt{\alpha}$. From perturbation theory, requiring $\alpha$ to be sufficiently small,  one then infers that \begin{equation}\label{1.11}
E_\mathrm{r}''(0) > 0.
\end{equation}
But for larger $\alpha$ it is difficult to exclude $E(P)$ to have vanishing curvature at  $P=0$. Under to be stated conditions we will establish the identity
\eqref{1.11}. Hence from $\sigma^2 >0$ one concludes that $m_\mathrm{eff} <\infty$ for all $\alpha$.

A related issue is the long-standing (physically obvious) conjecture
\begin{equation}\label{1.12}
E(P) > E(0),\qquad P \neq 0. 
\end{equation}
The weaker property (\ref{1.6}) follows from a Kato inequality for $H$ \cite{F74}. So the real issue is to exclude points $P^*\neq 0$ at which $E(P^*) = E(0)$.
Property \ref{1.12} is claimed in \cite{GL91},  Statement 2. In the proof on p. 78, the authors argue with the degeneracy of the ground state of $H$. But, under the common assumptions,
$H$ has no ground state at all. To have a ground state would require the set $\{P\in \real^d\,|\, E(P) = E(0)\}$ to have non-zero Lebesgue measure. 
In Section~\ref{General-analysis} we will explain how \eqref{1.12}  follows from a  CLT with yet to be studied boundary conditions. Alternatively one might invoke a suitable large deviation result in the context of the available boundary conditions.  

Our paper is organized as follows. In Section~\ref{sec2} we explain the connection to the probabilistic CLT and discuss recent results of interest in our context.
In Sections~\ref{polaron-type-section} and~\ref{Froehlich-section} we show the relations (\ref{1.10}), (\ref{1.11}) for polaron-type models with the UV cut-off  $\|g\|_{2} < \infty$ and for the Fr\"ohlich polaron, respectively. In Section~\ref{General-analysis} we study the functional analytic side of a CLT with general two-sided pinning. 

%%%%%%%%%%%%%%%%%%%%%%%%%%%%%%%%%%%%%%%%%%%%%%%%
\section{Probabilistic approach and central limit theorem}\label{sec2}
\setcounter{equation}{0} 
%%%%%%%%%%%%%%%%%%%%%%%%%%%%%%%%%%%%%%%%%%%%%%%%
We choose the boundary states $\phi_\pm = \varphi_\pm \otimes\Omega \in \mathcal{H}$ with $\Omega$ the Fock vacuum and define
\begin{equation}\label{2.1}
\tilde{G}_{T_-,T_+}(k,t) = \langle \phi_{-}, \mathrm{e}^{-T_-H}\mathrm{e}^{\mathrm{i}kx}\mathrm{e}^{-tH}\mathrm{e}^{-\mathrm{i}kx}\mathrm{e}^{-T_+H}\phi_+\rangle, 
\quad t,T_-,T_+ \geq 0.
\end{equation}
Using the direct inegral decompostion (\ref{1.3}), (\ref{1.4}), see also  formula~(\ref{Pi-on-state}), one obtains the identity
\begin{equation}\label{2.1a}
\tilde{G}_{T_-,T_+}(k,t) =  \int_{\real^ d} \mrm{d}P\,  \ov{\hat{\varphi}_-(P)}  \hat{\varphi}_+(P) \lan  \Omega, \e^{-T_-H(P)} \e^{-tH(P+k)} \e^{-T_+H(P)} \Omega\rangle. 
\end{equation}
In spirit of the Feynman-Kac formula for a Schr\"{o}dinger operator,  the semigroup $\mathrm{e}^{-tH}$, $ t \geq 0$, can be written as a weighted average with respect to a Gaussian measure. For the particle  trajectories we introduce the Wiener measure $\mathbb{P}^\mathrm{W}$ with expectation $\mathbb{E}^\mathrm{W}$. The continuous paths of the Wiener measure are denoted by $q(t)$. The Bose field maps to the Gaussian process $u(x,t)$ whose  
path measure is denoted by $\mathbb{P}^{\mathrm{G}}$, with expectation $\mathbb{E}^{\mathrm{G}}$. The Gaussian process has mean zero,
is stationary in space-time, and is uniquely defined through its covariance
\begin{equation}\label{2.2}
\mathbb{E}^{\mathrm{G}}  \big(u(x,t)u(x',t')\big) = \int_{\mathbb{R}^d} \mathrm{d} k \fr{1}{(2\pi)^d}  \frac{1}{2\omega(k)} \mathrm{e}^{\mathrm{i}k(x-x')}\mathrm{e}^{-\omega(k)|t-t'|}.
\end{equation}
Then 
\begin{eqnarray}\label{2.3}
&&\hspace{-30pt}\tilde{G}_{T_-,T_+}(k,t)  = \mathbb{E}^\mathrm{W}\!\!\times\! \mathbb{E}^\mathrm{G}\Big( \varphi_-(q(-T_-)) \varphi_+(q(T_++t)) \mathrm{e}^{-\mathrm{i}k(q(t) - q(0))} \nonumber\\
&&\hspace{50pt}\times\exp\Big[\sqrt{\alpha}\int_{-T_-}^{T_++t}\mathrm{d}s
\int_{\mathbb{R}^d} \mathrm{d}x\, v(x) u(q(s) - x,s) \Big]\Big),
\end{eqnarray}
which makes more explicit how $v(x)$ smears the field $u$ relative to the position of the particle.  If $\|g\|_2 < \infty$,
the term in the square brackets is a well-defined Gaussian random variable with respect to $\mathbb{P}^{\mathrm{G}}$. The Gaussian average $\mathbb{E}^\mathrm{G}$ can be carried out explicitly leading to  
\begin{eqnarray}\label{2.4}
&&\hspace{-30pt}\tilde{G}_{T_-,T_+}(k,t) =
\mathbb{E}^\mathrm{W} \Big(\varphi_-(q(-T_-))\varphi_+(q(T_++t))  \mathrm{e}^{-\mathrm{i}k(q(t) - q(0))}\nonumber\\
&&\hspace{50pt}\times\exp\Big[\tfrac{1}{2}\alpha \int_{-T_-}^{T_++t}\mathrm{d}s \int_{-T_-}^{T_++t}\mathrm{d}s'\,  W(q(s) - q(s'), s - s')\Big]\Big)
\end{eqnarray}
with 
\begin{equation}\label{2.5}
W(x,t) = \int _{\mathbb{R}^d}\mathrm{d}k|g(k)|^2\mathrm{e}^{\mathrm{i}kx} \mathrm{e}^{-\omega(k)|t|}.
\end{equation}
Note that $W$ is real, continuous, rotation invariant in $x$, and $|W(x,t)| \leq \|g\|^2_{2}$. In particular the integrand under the double time integral appearing in \eqref{2.4} is pathwise bounded and continuous. 

The Fr\"{o}hlich polaron is the special case $d=3$, $\omega = 1$, and $g(k) =( \sqrt{2\pi} |k|)^{-1}$, thereby defining the Hamiltonian $H^\mathrm{Fr}$, for which 
self-adjointness is established in \cite{F74,GW16}.
The kernel $W$ of the Fr\"{o}hlich polaron is given by 
\begin{equation}\label{2.6}
W^\mathrm{Fr}(x,t) =  |x|^{-1}\mathrm{e}^{-|t|},
\end{equation}
which is no longer bounded. Still the factor $\exp[\cdot]$ in \eqref{2.4} is integrable \cite{BT17}. To establish the validity of  the basic identity  \eqref{2.4} for $H^\mathrm{Fr}$ one introduces the cut-off coupling $g_{\ka}(k) = (  \sqrt{2}\pi  |k|)^{-1}\mathrm{e}^{-\fr{1}{\ka} |k|}$, thereby defining the
Hamiltonian $H_{\ka}$. The strong limit $\lim_{\ka \to \infty} \mathrm{e}^{-tH_{\ka} } =  \mathrm{e}^{-tH^\mathrm{Fr}}$ is established  in \cite{F74, Mi10}, 
which  controls the left side of \eqref{2.4}. On the right side $W$ is replaced by 
\begin{equation}\label{2.7}
W_\kappa(x,t) =   |x|^{-1}  \fr{2}{\pi}\arctan(\ka|x|)\mathrm{e}^{-|t|},
\end{equation}
which increases monotonously to $W^\mathrm{Fr}(x,t)$. Thus by monotonicity the right hand side of  \eqref{2.4} converges to the corresponding expression with kernel $W^\mathrm{Fr}(x,t)$
and hence  \eqref{2.4} remains valid for the Fr\"{o}hlich polaron.

In \eqref{2.4} the reference process is a standard Brownian motion over the time interval $[-T_-,t + T_+]$. The Brownian motion is pinned by the 
function $\varphi_-$ at the left border and by $\varphi_+$ at the right one. The Brownian path is weighted by the exponential of the double time integral involving $W$.
Note that the weight depends only on the increments. To have a probability measure we have to normalize by the partition function  $\tilde{G}_{T_-,T_+}(0,t)$. The difference $q(t) - q(0)$ is the Brownian motion increment over the time interval $[0,t]$. Of interest is its characteristic function, i.e. the Fourier transform of  the corresponding probability density function. Altogether this leads to the normalized characteristic function
\begin{equation}\label{2.8}
G_{T_-,T_+}(k,t) = \tilde{G}_{T_-,T_+}(k,t)/\tilde{G}_{T_-,T_+}(0,t).
\end{equation}
Depending on the precise set-up, one then has to establish the limits $T_-,T_+ \to \infty$ followed by the CLT which requires $t \to \infty$. 

In the probabilistic literature, 
two distinct boundary conditions have been studied and we discuss them one by one. In both cases $T_- =0$, $\varphi_-(x) = \delta(x)$, and
$\varphi_+(x) = 1$, of which the latter two have to be approximated by a suitable sequence of $L^2$ functions. We set $T_+ = T$ in the sequel. 

In \cite{BS05} and the follow-up by Gubinelli \cite{G06} the authors require the conditions
 \begin{equation}\label{2.9}
 \int _{\mathbb{R}^d}\mathrm{d}k|g(k)|^2\big(\sum_{j = 1,2,3} \omega^{-j}\big) < \infty,\quad
 \int _{\mathbb{R}^d}\mathrm{d}k|g(k)|^2|k|^2 \big(\sum_{j = 2,4} \omega^{-j}\big)<\infty.
\end{equation}
They consider $\ti G_{0,T}$ of the form
\beqa
\ti G_{0,T}(k,t)=\mathbb{E}^\mathrm{W} \Big(\delta(q(0))  \mathrm{e}^{-\mathrm{i}kq(t)}
\exp\Big[\tfrac{1}{2}\alpha \int_{0}^{T+t}\mathrm{d}s\int_{0}^{T+t}\mathrm{d}s'W(q(s) - q(s'), s - s')\Big]\Big)
\eeqa
and establish the limit
\begin{equation}\label{2.10}
\lim_{T \to \infty} G_{0,T}(k,t) = G_{0,\infty}(k,t).
\end{equation}
The CLT is proved, thus ensuring the limit
\begin{equation}\label{2.11}
\lim_{\epsilon \to 0} G_{0,\infty}(\epsilon k, \epsilon^{-2}t) =  \mathrm{e}^{-\frac{1}{2} \sigma^2 k^2 t}
\end{equation}
for some $\sigma >0$. In fact, the stronger functional CLT is established, see  \cite[Theorem 1.1]{BS05}. 

It is instructive to rewrite the expectation values from above in the language of operators as in (\ref{2.1}), (\ref{2.1a}), with the result  
\begin{eqnarray}\label{2.12}
&&\hspace{-46pt}\tilde{G}_{0,T}(k,t) = \langle \phi_{-}, \mathrm{e}^{\mathrm{i}kx}\mathrm{e}^{-tH}\mathrm{e}^{-\mathrm{i}kx}\mathrm{e}^{-TH}\phi_+\rangle\nonumber\\
&&\hspace{6pt}= \int_{\mathbb{R}^d}\mathrm{d}P\delta(P)  \langle \Omega, \mathrm{e}^{-tH(P+k)}\mathrm{e}^{-TH(P)}\Omega\rangle
= \langle \Omega, \mathrm{e}^{-tH(k)}\mathrm{e}^{-TH(0)}\Omega\rangle,
\end{eqnarray}
where we used $\hat{\varphi}_-(P) = (2\pi)^{-d/2}$, and $\hat{\varphi}_+(P) = (2\pi)^{d/2}\delta(P)$. For polaron-type models  treated in 
Section~\ref{polaron-type-section} below $H(0)$ has a spectral gap and a unique ground state $\psi_0$, thus by the spectral theorem
\begin{equation}\label{2.13}
G_{0,\infty}(k,t) = \langle \Omega, \mathrm{e}^{-t(H(k)- E(0))}\psi_0\rangle/\langle \Omega, \psi_0\rangle.
\end{equation}

More recently, Mukerjee and Varadhan studied the CLT under weaker conditions than imposed in  \cite{BS05,G06}. Their starting formula is
\begin{equation}\label{2.14}
\tilde{G}_{0,0}(k,t) =
\mathbb{E}^\mathrm{W} \Big(\delta(q(0))  \mathrm{e}^{-\mathrm{i}kq(t)}
\exp\Big[\tfrac{1}{2}\alpha \int_{0}^{t}\mathrm{d}s\int_{0}^{t}\mathrm{d}s'W(q(s) - q(s'), s - s')\Big]\Big),
\end{equation}
hence $T = 0$, which one recognizes as a particular case of (\ref{2.12}) and 
\begin{equation}\label{2.15}
G_{0,0}(k,t) = \langle \Omega, \mathrm{e}^{-tH(k)}\Omega\rangle/ \langle \Omega, \mathrm{e}^{-tH(0)}\Omega\rangle.
\end{equation}
In \cite[Theorem 4.2]{MV18} the CLT  of the following form is established for the Fr\"{o}hlich polaron, 
\beqa
\lim_{\epsilon \to 0} G_{0,0}(\epsilon k, \epsilon^{-2}t) =  \mathrm{e}^{-\frac{1}{2} \sigma^2 k^2t}, \label{central-limit-Froehlich}
\eeqa
for some $\si>0$, with the restriction  $\alpha\in [0,\alpha_0)\cup (\alpha_1,\infty)$ for some $0<\alpha_0<\alpha_1<\infty$. %\in (0,\infty)$,

The functional CLT is not touched upon. In the related study \cite{MV18a} the strong coupling limit and its relation to the Pekar process are investigated. 
Mukherjee \cite{M19}    also starts from \eqref{2.14} and considers a general weight function $W$, for which he requires
$|W(x,t)| \leq C (1 + |t|)^{-(2 +\delta)}$ for some $C, \delta > 0$. In particular, this condition covers the polaron 
whenever $g \in L^2$. [ In the currently posted version in addition $W \geq 0$ is required. As communicated to us 
by the author this condition can be dropped.] In \cite[Theorem 2.1]{M19} the conventional CLT of the form (\ref{central-limit-Froehlich}) is proved for arbitrary
$\al\geq 0$. 

Physically one is also interested in the behavior of $E(P)$ away from the origin. Starting from \eqref{2.15}, instead of $k = \mathcal{O}(\epsilon)$ and would have to consider $k = P = 
\mathcal{O}(1)$, which probabilistically is a problem of large deviations. In Section~\ref{General-analysis} we explore a different approach by starting from \eqref{2.1} with general 
square-integrable boundary functions $\varphi_\pm$ in the limit $T_\pm \to \infty$, but still invoking a CLT.

%%%%%%%%%%%%%%%%%%%%%%%%%%%%%%%%%%%%%%%%%%%%%%%%%%%%%%%%%%%
\renewcommand{\Phi}{\phi}
\renewcommand{\Psi}{\psi}
%%%%%%%%%%%%%%%%%%%%%%%%%%%%%%%%%%%%%%%%%%%%%%%%%%%%%%%%%%%%%
\section{Polaron-type models with a UV cut-off} \label{polaron-type-section}
\setcounter{equation}{0}
%%%%%%%%%%%%%%%%%%%%%%%%%%%%%%%%%%%%%%%%%%
In this section we show that $\si>0$ appearing in the CLT  (\ref{central-limit-Froehlich})  coincides with the
effective mass for  a large class of polaron-type Hamiltonians with a UV cut-off.
It is convenient to start from a family of the fiber Hamiltonians of the form
\beqa
H(P)=\h (P-P_{\mrm{f}})^2+H_{\mrm{f}}+ \sqrt{\alpha} \int _{\mathbb{R}^d}\mathrm{d}k\,g(k)\big(a(k) + a^*(k)\big), \label{fiber-hamiltonians}
  \eeqa
where $H_\mrm{f}= \int _{\mathbb{R}^d}\mathrm{d}k\, \omega(k) a^*(k) a(k)$, $P_{\mrm{f}}= \int _{\mathbb{R}^d}\mathrm{d}k\, k\, a^*(k) a(k)$. 
Further asumptions are listed~in\medskip\\
\textbf{Condition C}. (i) $g \in L^2(\real^d)$ is real and rotation invariant.  The coupling constant $\al\geq 0$ is arbitrary.\medskip\\
(ii) $\omega(k) \geq c_0 >0$, $\omega$ is continuous, rotation invariant, and sub-additive in the sense that
\beqa
\omega(k_1+k_2) \leq \omega(k_1) + \omega(k_2).
\eeqa
Then, by the Kato-Rellich theorem, $H(P)$ are self-adjoint, semi-bounded operators on the  domain $D(P_{\mrm{f}}^2+H_{\mrm{f}})$ which is independent of $P$.
By the  direct integral formula (\ref{1.3})  one obtains 
a Hamiltonian of the form (\ref{1.1}). 
Under the above assumptions, the HVZ theorem for these models was shown in \cite{F74,SM14}.  
All the properties below can be found in \cite{SM14} except for part 0 for which we refer to \cite{Gro72} or \cite[Section 15.2]{Sp},
and part 6 which can be found in \cite{MR12}.
We refer to \cite{SM14} for  a discussion of the literature.
%%%%%%%%%%%%%%%%%%%%%%%%%%%%%%%
\begin{lemma}\label{Jacob}\emph{\cite{SM14,F74}} Assume Condition $C$ 
and define $E(P)=\inf\, \mrm{spec}(H(P))$, $E_{\mrm{ess}}(P)=\inf \mrm{spec}_{\mrm{ess}}(H(P))$.
Then the following statements hold true:
\begin{enumerate}

\item[0.]  $E(0)\leq E(P)$ for all $P\in \real^d$.

\item[1.]  $E_{\mrm{ess}}(P)=\mrm{inf}_{k\in \real^d} (E(P-k)+\om(k))$.

\item[2.] The interval $\mathcal{I}_0:=\{\,P\in \real^d\,|\, E(P)<E_{\mrm{ess}}(P)\,\}$ 
 is non-empty and contains a neighbourhood of any global minimum of $|P|\mapsto E_{\mrm{r}}(|P|)$.

\item[3.]  $E(P)$ is an isolated, simple eigenvalue for $P\in \mathcal{I}_0$.

\item[4.]  If $\om$ is bounded, then, for any $\ov{E}\in \real$, $\lim_{|P|\to \infty, E(P)\leq \ov{E}} (E_{\mrm{ess}}(P)- E(P))=0$.\\
In general, $E(P)\geq c_1\om(P)+c_2$ for some  $c_1>0$.

\item[5.] For $P\in \mathcal{I}_0$ we have $|\lan \Om , \psi_P\ran|>0$, where $\psi_P$ is the ground state of $H(P)$.

\item[6.] $\mathcal{I}_0\ni P\mapsto E(P)$ and $P\mapsto |\psi_P\ran \lan \psi_P|$ are real analytic functions.

\end{enumerate}

\end{lemma}
%%%%%%%%%%%%%%%%%%%%%%%%%%%%%%%%%%%%
Let us comment briefly on the proof of  properties 1--5 and the role of various assumptions on $\om$. 
We consider the thresholds 
\beqa
E^{(n)}(P):=\inf_{k_1, \ldots, k_n\in \real^d} (E(P-k_1-\cdots-k_n)+\om(k_1)+\cdots+\om(k_n)).
\eeqa
Assuming only that $\om$ is continuous, bounded and massive (i.e. $\omega(k) \geq c_0 >0$),  Theorem~2.1 of \cite{SM14} gives $E_{\mrm{ess}}(P)=\inf_{n\geq 1}   E^{(n)}(P)$.
As sub-additivity of $\om$ clearly gives monotonicity of thresholds, with this additional assumption one obtains property 1~of Lemma~\ref{Jacob}
above. As pointed out in \cite{SM14}, it is clear from this relation, and from the fact that $\om$ is massive, that if   $E(P)< \inf_{P'} E(P')+\inf_{k'}\om(k')$ 
then $E(P)< E_{\mrm{ess}}(P)$, which gives property 2  of Lemma~\ref{Jacob}.
Clearly, the spectrum below $ E_{\mrm{ess}}(P)$ consists
at most of eigenvalues of finite multiplicity with $E_{\mrm{ess}}(P)$ as the only possible accumulation point. Thus  $E(P)$ is an isolated 
eigenvalue, which  is simple by Theorem 2.3 of \cite{SM14}. Thus we obtain property 3  of Lemma~\ref{Jacob}. For the first part of property 4 and property 5 
we refer to Theorems 2.3 and 2.4 of \cite{SM14}. The second part of property 4 can be found in \cite{F74} (see also \cite[Section~15.2, property (v)]{Sp}).

As for part 6, 
we note that for $\xi$ in the resolvent set of $H(P_0)$ the function $P\mapsto (H(P)-\xi)^{-1}$ can be expanded around any $P_0\in \real^d$ as in formula~(\ref{expansion-in-P}).  
The real analyticity of  the  eigenprojections $\mathcal{I}_0\ni P\mapsto |\Psi_P\ran \lan \Psi_P|$ follows immediately via the Cauchy formula. (We note that
by a suitable choice of the phase, we can ensure that $\mathcal{I}_0\ni P\mapsto \psi_P$ is norm-continuous, which is the property we will need below).
Since $|P|\mapsto H(\hat{P} |P|)$ is a real analytic family of self-adjoint operators in the sense of \cite[Chapter VII, \S 1, \S 3]{Ka} and $P\mapsto E(P)$ is a rotation invariant function, we obtain by \cite[Chapter VII, \S 3]{Ka} that $\mathcal{I}_0\ni P\mapsto E(P)$ is real analytic.

Now we are ready to state and prove our main result concerning polaron-type models with a UV cut-off. 
%%%%%%%%%%%%%%%%%%%%%%%%%%%%%%%
\begin{theorem}\label{corollary-polaron-type}
Consider polaron-type models satisfying Condition C. 
Then, for all $\alpha\geq 0$,  
\begin{equation}\label{4.10}
(m_\mathrm{eff})^{-1} = \sigma^2 >0.
\end{equation}
\end{theorem}
%%%%%%%%%%%%%%%%%%%%%%%%%%%%%%%
\begin{proof}  The proof relies on the CLT as stated in (\ref{central-limit-Froehlich}). 
We consider the expression
\beqa
G_{0,0}(\eps k,\eps^{-2}t)=\fr{ \langle \Omega, \mathrm{e}^{-\fr{t}{\eps^2} (H(\eps k)-E(0)) }\Omega\rangle}{\langle \Omega, \mathrm{e}^{-\fr{t}{\eps^2} (H(0)-E(0)) } \Omega\rangle}.
\label{G-0-0}
\eeqa
The spectral calculus immediately gives $\lim_{\eps\to 0}\langle \Omega, \mathrm{e}^{-\fr{t}{\eps^2} (H(0)-E(0)) } \Omega\rangle=\lan \Om, \psi_0\ran \lan \psi_0, \Om\ran$. 
 Concerning the numerator in (\ref{G-0-0}), we obtain
\begin{align}
 \langle \Omega, \mathrm{e}^{-\fr{t}{\eps^2} (H(\eps k)-E(0)) }\Omega\rangle&= \langle \Omega, \mathrm{e}^{-\fr{t}{\eps^2} (H(\eps k)-E(0)) }\psi_{\eps k}\ran 
 \lan \psi_{\eps k}, \Omega\ran \label{leading-term} \\
&\ph{44}+\langle \Omega, \mathrm{e}^{-\fr{t}{\eps^2} (H(\eps k)-E(0)) }(|\psi_{\eps k}\ran 
 \lan \psi_{\eps k}|)^{\bot} \Omega\ran\label{sub-leading} \\
&\ph{44} \underset{\eps\to 0}{\to}  \lan \Om, \psi_0\ran \lan \psi_0, \Om\ran    \e^{-\h tk^2 (\pa^2_{|P|} E_{\mrm{r}})(0)  },
\end{align}
where in the leading term (\ref{leading-term}) we used the analyticity of $P\mapsto E(P)$ near zero (see Lemma~\ref{Jacob})
and we noted that the expression in  (\ref{sub-leading}) tends to zero as $ \eps \to 0$ by the spectral calculus.
\end{proof}
%%%%%%%%%%%%%%%%%%%%%%%%%%%%%%%%%%%%%%%%%%%%%%%%%%%%%%%

We remark that a priori the diffusion constant obtained from the CLT of the characteristic function \eqref{2.15}  could differ from the one of 
\eqref{2.11}.   Our result implies that they agree under Condition C and assumption (\ref{2.9}).
%Our result implies that they agree under Condition C {\bc and assumption (\ref{2.9}) not needed, only $L^2$}.
\newpage

%%%%%%%%%%%%%%%%%%%%%%%%%%%%%%%%%%%%%%
\section{The Fr\"ohlich polaron} \label{Froehlich-section}
\setcounter{equation}{0}
%%%%%%%%%%%%%%%%%%%%%%%%%%%%%%%%%%%%%%
Let $H_{\ka}(P)$ be the polaron Hamiltonian (\ref{fiber-hamiltonians}) with  $d=3$, $\om\equiv 1$  and $g(k)=\fr{ \chi_{[0,\ka]}(|k|)} {  \sqrt{2} \pi |k|}$, where $\ka$ is the UV cut-off\footnote{We use
here a different UV cut-off than in the discussion in Section~\ref{sec2}. However, the limiting Fr\"ohlich Hamiltonians $H^{\mrm{Fr}}(P)$ are the same, as one can infer
from \cite[Proposition~A.4]{Mi10} and the strong convergence  of the Gross transform.}.
Explicitly, it has the form  
\beqa
H_{\ka}(P)=\h (P-P_{\mrm{f}})^2+N_{\mrm{f}}+ \sqrt{\alpha} \int _{|k|\leq \ka}\mathrm{d}k \fr{1}{   \sqrt{2} \pi  |k|}\big(a(k) + a^*(k)\big), \label{cut-off-Hamiltonians}
\eeqa
where $N_{\mrm{f}}$ is the number operator. It is well known, that this sequence of Hamiltonians converges in the norm resolvent sense
as $\ka\to \infty$ to the limiting Fr\"ohlich Hamiltonian $H^{\mrm{Fr}}(P)$.  Also, the sequence of the full Hamiltonians 
$H_{\ka}=\Pi^*\int^{\oplus} \mrm{d}P\, H_{\ka}(P) \Pi$ converges in the norm resolvent sense to $H^{\mrm{Fr}}=\Pi^*\int^{\oplus}\mrm{d}P H^{\mrm{Fr}}(P) \,\Pi$, cf. 
\cite{GW16}  and references therein.
 Making use of these approximation properties, Lemma~\ref{Jacob} and further results from \cite{Mi10}, it is easy to establish the
 following:
%%%%%%%%%%%%%%%%%%%%%%%%%%%%%%%
\begin{lemma}\label{Froehlich-theorem}  Let $E(P)=\inf\, \mrm{spec}(H^{\mrm{Fr}}(P))$  and $E_{\mrm{ess}}(P)=\inf \mrm{spec}_{\mrm{ess}}(H^{\mrm{Fr}}(P))$.
Then the following statements hold true:
\begin{enumerate}
\item[0.] $E(0)\leq E(P)$ for all $P\in \real^3$.

\item[1.]  $E_{\mrm{ess}}(P)=E(0)+1$.  

\item[2.] The interval $\mathcal{I}_0:=\{\,P\in \real^d\,|\, E(P)<E_{\mrm{ess}}(P)\,\}$ 
  contains a neighbourhood of any global minimum of $|P|\mapsto E_{\mrm{r}}(|P|)$. 
 
\item[3.]  $E(P)$ is an isolated, simple eigenvalue for $P\in \mathcal{I}_0$.

\item[4.]  All global minima of $|P|\mapsto E_{\mrm{r}}(|P|)$ are contained in a compact set.

\item[5.] For $P\in \mathcal{I}_0$ we have $|\lan \Om , \psi_P\ran|>0$, where $\psi_P$ is the ground state of $H(P)$.

\item[6.]  $\mathcal{I}_0\ni P\mapsto E(P)$ and $P\mapsto |\psi_P\ran \lan \psi_P|$ are real analytic functions.

\end{enumerate}

\end{lemma}
%%%%%%%%%%%%%%%%%%%%%%%%%%%
Let us comment on the proofs of the above properties. It is a general consequence of the strong resolvent convergence
that for any eigenvalue $\la$ of $H(P)$ there exists an approximating sequence $\la_{\ka}\to \la$ of eigenvalues
of $H_{\ka}(P)$ \cite[Theorem VIII.24]{RSI}. Therefore, part~0 of Lemma~\ref{Froehlich-theorem}  follows from part 0 of Lemma~\ref{Jacob}.
Next, by \cite[Proposition A.4]{Mi10}, $\lim_{\ka\to \infty} E_{\mrm{ess},\ka}(P)=E_{\mrm{ess}}(P)$, where 
$E_{\mrm{ess},\ka}(P)$ is the bottom of the essential spectrum of $H_{\ka}(P)$. Now part 1 of Lemma~\ref{Froehlich-theorem}
follows from part 1 of Lemma~\ref{Jacob} applied to the case of  $\om\equiv 1$. (Alternatively, one can refer to \cite[Section IV]{S88}).  Parts 2 and 3 of Lemma~\ref{Froehlich-theorem}
follow  from parts 1 and 2 of the same {\rc lemma}, {\rc considering that the proof of \cite[Theorem 6.4]{Mi10} gives the uniqueness of the ground state 
whenever it exists, also outside of the ball $|P|<\sqrt{2}$ from the statement of the theorem}. Concerning part 4, suppose by contradiction that there is a sequence
$P_{\ell}$, $\ell\in \nat$, s.t. $E(P_{\ell})=E(0)$ and $|P_{\ell}|\to \infty$. We pick a function $f\in C_0^{\infty}(\real)$ supported in
a ball around $E(0)$ of radius strictly smaller than $1$ and s.t.   $0\leq f\leq 1$ and  $f(E(0))=1$. Then, by the norm resolvent convergence of $H_{\ka}$ and \cite[Theorem VIII.20]{RSI} we have
\begin{align}
0&=\lim_{\ka\to \infty}\| f(H_{\ka} )-f(H^{\mrm{Fr}})\|=\lim_{\ka\to \infty} \sup_{P\in \real^3}\| f( H_{\ka}(P) )-f(H^{\mrm{Fr}}(P) )\|\non\\
 &\geq \lim_{\ka\to \infty} \sup_{\ell \geq \ell_{\ka}  }\| f(H_{\ka}(P_{\ell}) )-f(H^{\mrm{Fr}}(P_{\ell}) )\|=1,
 \end{align}
 which is a contradiction. Here in the third step we choose $\ell_{\ka}$ so large that the spectrum of  $H_{\ka}(P_{\ell})$ is outside of the support
 of $f$, which is possible by Lemma~\ref{Jacob}, part 4.  Part 5 of Lemma~\ref{Froehlich-theorem} is a consequence of the strict positivity statement in \cite[Theorem~6.4]{Mi10}, {\rc where again  we can disregard the restriction $|P|<\sqrt{2}$, considering the structure of the proof}.
 Part 6 is {\rc proven analogously as the corresponding part of Lemma~\ref{Jacob}, given the input from}  Appendix~\ref{holomorphy}. 

Now we come to our main result concerning the Fr\"ohlich polaron.
%%%%%%%%%%%%%%%%%%%%%%%%%%%%%%%%%%%%%%%%%%%%%%%%%%%%%%%%%%%
\begin{theorem} \label{Froehlich-corollary} Consider the Fr\"ohlich polaron. 
Then, for all $\alpha\in [0,\alpha_0)\cup (\alpha_1,\infty)$ for some $0<\alpha_0<\alpha_1< \infty$,  
\begin{equation}\label{4.10}
(m_\mathrm{eff})^{-1} = \sigma^2 >0.
\end{equation}
\end{theorem}
%%%%%%%%%%%%%%%%%%%%%%%%%%%%
\begin{proof} The claim follows from the CLT stated in (\ref{central-limit-Froehlich}) by the same
steps as in the proof of Theorem~\ref{corollary-polaron-type}. Instead of Lemma~\ref{Jacob}, Lemma~\ref{Froehlich-theorem} is used. 
\end{proof}

%%%%%%%%%%%%%%%%%%%%%%%%%%%%%%%%%%%%%%%%%%%%%%%%%%%%%%%%%%%%%%%%
 \section{A CLT for two-sided pinning}\label{General-analysis}
\setcounter{equation}{0}
%%%%%%%%%%%%%%%%%%%%%%%%%%%%%%%%%%%%%%%%%%%%%%%%%%%%%%%%%%%%%%%%

We return to the set-up  of equation \eqref{2.1} with square-integrable boundary functions $\varphi_\pm = \varphi$, $T_\pm = T$, and $T \to \infty$. 
A probabilistic study of this variant does not seem to be available in the literature and we focus on the functional analytic side. It will be more transparent to work in a  general framework,
which includes the polaron models discussed so far, but many more, e.g. systems with a non-quadratic energy momentum relation for the electron. 

Let $\hil$, $\mcFn$ be Hilbert spaces and $\Pi: \hil\to L^2(\real^d, \mcFn)$  a unitary.   For any $\Phi\in \hil$ 
we have the corresponding representation $\Phi=\Pi^*\int^{\oplus}_{\real^d} \mrm{d}P\,\Phi_P$. For any $k\in \real^d$
we define the unitary $U(k)$ by its action on such vectors $\phi$
\beqa
U(k)\Phi= \Pi^*\int^{\oplus}_{\real^d} \mrm{d}P\,\Phi_{P+k}.
\eeqa
Furthermore, we are interested in self-adjoint operators $H$ on a domain  $D(H)\subset \hil$ which have the representation
\beqa
H=\Pi^{*} (\int^{\oplus}_{\real^d} \mrm{d}P\, H(P)) \Pi. \label{fiber-decomposition}
\eeqa
Here $\real^d\ni P\mapsto H(P)$ is a real analytic family of positive operators with domains $D(H(P))\subset \mcFn$, as stated more precisely in the standing assumption 0 below.
 Furthermore, we note that for any bounded Borel function~$f$ 
\beqa
U(k)f(H)U(k)^*=\Pi^{*} (\int^{\oplus}_{\real^d} \mrm{d}P\, f(H(P+k)))\Pi. \label{A0}
\eeqa

In this section we impose the following \textbf{standing assumptions}:
\begin{enumerate}

\item[0.]  The family $P\mapsto H(P)$ is real analytic in the sense that for any $P_0\in \real^d$ and  any $\xi\notin \mrm{spec}(H(P_0))$ there 
exists a real neighbourhood $N_{P_0}$ of $P_0$ s.t.   $\xi\notin \mrm{spec}(H(P))$   for any $P\in N_{P_0}$ and $N_{P_0}\ni P\mapsto (H(P)-\xi)^{-1}$
is real analytic.  As a consequence, for any $\hat{P}$ on the unit sphere $|P|\mapsto H(\hat{P} |P|)$ is a real analytic family of unbounded operators in the sense of  \cite[Chapter VII, \S 1]{Ka}.
Another consequence of this property and of the {\rc Helfer-Sj\"ostrand method of almost analytic extensions} 
%formula $\e^{-x}=\lim_{n\to\infty}(1+x/n)^{-n}$ 
is the strong continuity of  $\real^d\ni P\mapsto \e^{-tH(P)}$, which will be used in the proofs below.

\item[1.]   The function $P\mapsto E(P):=\mathrm{infspec} (H(P))$ is rotation invariant and we write as before $E(P)=E_{\mrm{r}}(|P|) $. $E$ attains its global minima in the sets
$M_{\ell}=\{\,P\in \real^d\,|\, |P|=Q_{\ell}\, \}$, $\ell=0,1,2\ldots, L$, where $0\leq Q_0<Q_1<\ldots< Q_{L}$ and $L$ finite. Also, we assume $E_{\mrm{r}}(Q_{\ell} )=0$.

\item[2.]  $E$ is analytic in sets $\ti{M}_{\ell}=\{\,P\in \real^d\,|\, |P|\in \De_{\ell}\, \}$, where $\De_{\ell}$ is a neighbourhood of $Q_{\ell}$.
 Then we have  $E_{\mrm{r}}(  {\rc Q_{\ell}+R})\sim R^{n_{\ell}}$ for small $R$ and some $n_{\ell}\in \nat$, $n_{\ell}\geq 2$.

\item[3.]  For $P\in \ti{M}_{\ell}$,  $E(P)$ are simple eigenvalues and the corresponding family of projections $\ti{M}_{\ell}\ni P\mapsto |\Psi_{P}\ran \lan \Psi_P|$ 
is strongly continuous. (It easily follows that $P\mapsto \Psi_P$ can be chosen strongly continuous in $\ti{M}_{\ell}$ by a suitable choice of the phases, which is what will be used below).

\item[4.] There exist vectors  $\phi\in \hil$ such that $\ti{M}_{\ell}\ni P\mapsto |\lan \phi_P, \psi_P\ran|$ are continuous and non-zero on $M_{\ell}$.

\end{enumerate}
The above assumptions hold, in particular, for models of Sections~\ref{polaron-type-section} and \ref{Froehlich-section} as shown in
the following proposition. The proof is postponed to Appendix~\ref{proposition}.
%%%%%%%%%%%%%%%%%%%%%%%%%%%%%%%%%%%%%%%%%%%%%%%%%%%%%%%%%%%%%%%%%%%%%
\begin{proposition}\label{corollary}  For the polaron-type models (\ref{fiber-hamiltonians}) satisfying  Condition C and the Fr\"ohlich polaron~(\ref{cut-off-Hamiltonians}) the following properties hold true:
\begin{enumerate}
\item[(a)]  The models  satisfy the standing assumptions 0,1,2,3
above. 
\item[(b)] Let  $\phi=\varphi \otimes\Omega\in \hil$ be s.t. $\varphi\in L^2(\real^d)$, $\hat{\varphi}\in C(\real^d)$ and $\hat{\varphi}(p)>0$ for  all $p\in \mathcal{I}_0$. 
For such $\phi$ assumption 4 above holds. 
\end{enumerate}
\end{proposition}
%%%%%%%%%%%%%%%%%%%%%%%%%%%%%%%%

Coming back to the general framework, we note  that assumptions 2, 3 follow from 0,1 and  analytic perturbation theory \cite[Chapter VII, \S 3]{Ka} if $E(P)$ are simple, isolated eigenvalues for $P\in \ti{M}_{\ell}$. However, our discussion in this section does not require spectral gaps above $E(P)$. Also 
the standard relation $E(0)\leq E(P)$, $P\in \real^d$, for polaron type models, cf. Lemmas~\ref{Jacob}, \ref{Froehlich-theorem}, does not follow from the standing assumptions above. However,  with  additional input which we now explain, we will obtain not only this relation, but even $E(0)<E(P)$, $P\neq 0$, for $d\geq 2$.    

For the two-sided boundary condition the properly  normalized  characteristic  function reads 
\beqa
\ti{G}_{T}(k,t):= \lan \phi, \e^{-TH} U(k)\e^{-tH}U(k)^*\e^{-TH}\phi\ran, \quad G_{T}(k,t) := \tilde{G}_T(k,t)/\tilde{G}_T(0,t)  \label{A}
\eeqa
for $\phi\in \hil$ and $t,T\geq 0$.  By  the spectral theorem, the denominator above is different from zero for any finite $T$.
Furthermore, for $\phi$ as in assumption 4, the limits 
\beqa
G_{\infty}(k,t):=\lim_{T\to \infty} G_{T}(k,t) \quad\mrm{and}\quad  \lim_{\epsilon \to 0} G_{\infty}(\epsilon k, \epsilon^{-2}t)
\eeqa
exist under our standing assumptions. The explicit expressions  are provided in  Proposition~\ref{main-proposition}  and Lemma~\ref{epsilon-limit} below.
 We expect that the latter limit has the form suggested by the CLT.
%%%%%%%%%%%%%%%%%%%%%%%%%%%%%
\begin{con}\label{CLT-conjecture} 
There exists  $\phi\in \hil$ as in assumption 4 above, such that the CLT of the form
\begin{equation}\label{implicit-CLT}
\lim_{\epsilon \to 0} G_{\infty}(\epsilon k, \epsilon^{-2}t) =  \mathrm{e}^{-\frac{1}{2} \sigma^2 k^2 t}
\end{equation}
holds true for some $\si>0$.
\end{con}
%%%%%%%%%%%%%%%%%%%%%%%%%%%%%
The consequences of this conjecture for models satisfying the above standing assumptions are collected in the
following theorem.
%%%%%%%%%%%%%%%%%%%%%%%%%%%%%%%%%%%%%%%%%
\begin{theorem} \label{main-general-theorem} Suppose that Conjecture~\ref{CLT-conjecture} holds true for some $\phi\in \hil$
as in assumption~4 and $\si>0$.
Then, for $d\geq 2$, there is a  global minimum at zero (i.e. $Q_0=0$). Furthermore,
\begin{enumerate}
\item[(a)] $\si^2=(\pa^2_{|P|} E_{\mrm{r}})(0)$,
\item[(b)]  $E(P)>E(0)$ for $P\neq 0$.
\end{enumerate}
For $d=1$ we obtain that $\si^2=(\pa^2_{|P|} E_{\mrm{r}})(Q_{\ell})$ for $\ell=0,1,2,\ldots L$ and $0\leq Q_0<Q_1<\ldots <~Q_L$.
\end{theorem}
%%%%%%%%%%%%%%%%%%%%%%%%%%%%%%%%%%%%%%%%%%%%%%
  We stated a minimal conjecture as required for Theorem \ref{main-general-theorem} to hold. In fact, the CLT should be in force for a large set of boundary functions, e.g. those satisfying 
4. of the standing assumptions. Our theorem then asserts that the diffusion constant is always given by $\si^2=(\pa^2_{|P|} E_{\mrm{r}})(0)$.

The observation behind Theorem \ref{main-general-theorem} is fairly elementary and can be grasped most easily for the polaron models underlying  \eqref{2.1a}. We note that the $P$-integral has the weight
$ |\hat{\varphi}(P)|^2  > 0$ for $\varphi$ as in {Proposition~\ref{corollary}(b). Thus in the limit $T \to \infty$ the $P$-integral concentrates on the set 
of global minima $\{P \in \real^d\,|\, E(P) = E(0) \}$.  If the CLT would hold,
the limit expression must have come only from $P=0$ and hence $E(P) > E(0)$ for $P\neq0$. 

The actual proof is more involved  and the remaining part of this section is devoted to proving Theorem~\ref{main-general-theorem}.
We start with two auxiliary results, which do not rely on Conjecture~\ref{CLT-conjecture}.
%%%%%%%%%%%%%%%%%%%%%%%%%%%%%%%%
\begin{proposition} \label{main-proposition} The following statements hold:
\begin{enumerate}
\item If $Q_0=0$ and it is the  only global minimum of $E$, then
\beqa
\lim_{T\to \infty}  G_{T}(k,t)     =\lan \Psi_0, \e^{-tH(k)}   \Psi_0\ran. \label{one-minimum-formula}
\eeqa 
\item  If  $Q_0=0$ and there are other global minima at $Q_{\ell}> 0$, $\ell=1,2, \ldots, L$, we set $n:=\mrm{max}_{\ell\neq 0}(n_{\ell})$ {\rc (see assumption 2)}  and distinguish the following cases:  
\begin{enumerate}

\item[(a)] For $n>\fr{n_0}{d}$
\beqa
\lim_{T\to\infty}   G_{T}(k,t)   =
\fr{ \sum_{\ell}C_{\ell} \int \mrm{d}\Om(\hat{P}) | \lan \Phi_{Q_{\ell}\hP},\Psi_{Q_{\ell}\hP}\ran|^2 \lan \Psi_{Q_{\ell}\hP}, \e^{-tH({Q_{\ell}\hP}+k}) \Psi_{Q_{\ell}\hP}\ran}
{\sum_{\ell}C_{\ell} \int \mrm{d}\Om(\hat{P}) | \lan \Phi_{Q_{\ell}\hP},\Psi_{Q_{\ell}\hP}\ran|^2  }, \label{a}
\eeqa
where  $C_{\ell}\neq 0$ and the sums extend only over $\ell>0$ s.t.   $n_{\ell}=n$.

\item[(b)] For $n=\fr{n_0}{d}$
\begin{align}
 &\lim_{T\to\infty}   G_{T}(k,t)  \non\\
 &=\fr{c_{0}  | \lan \Phi_{0},\Psi_{0}\ran|^2  \lan \Psi_0, \e^{-tH(k)}   \Psi_0\ran+ \sum_{\ell}c_{\ell} \int \mrm{d}\Om(\hat{P}) | \lan \Phi_{Q_{\ell}\hP},\Psi_{Q_{\ell}\hP}\ran|^2 \lan \Psi_{Q_{\ell}\hP}, \e^{-tH({Q_{\ell}\hP}+k}) \Psi_{Q_{\ell}\hP}\ran}
{ c_{0}  | \lan \Phi_{0},\Psi_{0}\ran|^2+  \sum_{\ell}c_{\ell} \int \mrm{d}\Om(\hat{P}) | \lan \Phi_{Q_{\ell}\hP},\Psi_{Q_{\ell}\hP}\ran|^2  },\non\\ \label{b}
\end{align}
where $c_0,c_{\ell}\neq 0$ and the sums extend only over $\ell>0$ s.t.   $n_{\ell}=n$.

\item[(c)] For $n< \fr{n_0}{d}$
\beqa
\lim_{T\to \infty}  G_{T}(k,t)     =\lan \Psi_0, \e^{-tH(k)}   \Psi_0\ran. \label{c}
\eeqa 

\end{enumerate}

\item If $0<Q_0<Q_1<\ldots <Q_L$, for $L\geq 0$,  we obtain 
\beqa
\lim_{T\to\infty}   G_{T}(k,t)   =
\fr{ \sum_{\ell}C_{\ell} \int \mrm{d}\Om(\hat{P}) | \lan \Phi_{Q_{\ell}\hP},\Psi_{Q_{\ell}\hP}\ran|^2 \lan \Psi_{Q_{\ell}\hP}, \e^{-tH({Q_{\ell}\hP}+k}) \Psi_{Q_{\ell}\hP}\ran}
{\sum_{\ell}C_{\ell} \int \mrm{d}\Om(\hat{P}) | \lan \Phi_{Q_{\ell}\hP},\Psi_{Q_{\ell}\hP}\ran|^2  }, \label{a1}
\eeqa
where  $C_{\ell}\neq 0$ and  the sum extends over $\ell$ s.t. $n_{\ell}=\bar{n}:=\max_{\ell'}{n_{\ell'}}$.

\end{enumerate}

For $d=1$ the angular integrations above amount to summation over $\hat P=\pm 1$.

\end{proposition}
%%%%%%%%%%%%%%%%%%%%%%%%%%%%%%%%%
\begin{proof}  In the fiber representation expression~(\ref{A}) has the following form
\beqa
 \ti{G}_{T}(k,t) = \int_{\real^ d} \mrm{d}P  \lan \Phi_P, \e^{-TH(P)} \e^{-tH(P+k)} \e^{-TH(P)} \Phi_P\ran. 
\eeqa
We denote the spectral measure of $H(P)$ by $M_P(\,\cdot\,)$ and choose $\de>0$ s.t. $E(P)\leq \de$ implies that $P\in \ti{M}:=\bigcup_{\ell=0}^{L}\ti{M}_{\ell}$.
By spectral calculus, we have
\beqa
\mrm{n}\textrm{-}\mrm{lim}_{T\to \infty}\e^{-TH(P)} M_P((\de, \infty))=0, \quad \mrm{s}\textrm{-}\mrm{lim}_{T\to \infty}\e^{-TH(P)} M_P((E(P), \infty))=0.
\eeqa
Therefore, it suffices to study 
\begin{align}
\ti{G}^{(1)}_{T}(k,t) &=   \int_{  \ti{M}  } \mrm{d}P  \lan \Phi_P,\Psi_P\ran \lan \Psi_P, \e^{-TH(P)} \e^{-tH(P+k)} \e^{-TH(P)} \Psi_P\ran \lan \Psi_P,\Phi_P\ran \non\\
&=\int_{ \ti{M} }\mrm{d}P \,| \lan \Phi_P,\Psi_P\ran|^2      \e^{-2TE(P)}    \lan \Psi_P, \e^{-tH(P+k)}  \Psi_P\ran \non\\
&= \int_{ \ti{M}  } \mrm{d}P \,| \lan \Phi_P,\Psi_P\ran|^2      \e^{-(2T+t)E(P)}    \lan \Psi_P, \e^{-tH(P+k)} \e^{tH(P)}  \Psi_P\ran.
\end{align}
Hence, setting $  {G}^{(1)}_{T}(k,t):=  \ti{G}^{(1)}_{T}(k,t)/ \ti{G}^{(1)}_{T}(0,t)$,
\begin{align}
{G}^{(1)}_{T}(k,t)=  \int_{\ti{M}} \mrm{d}P \, \bigg\{  \fr{| \lan \Phi_P,\Psi_P\ran|^2      \e^{-(2T+t)E(P)} }{  \int_{\ti{M}} \mrm{d}P'  | \lan \Phi_{P'},\Psi_{P'}\ran|^2  \e^{-(2T+t)E(P')} } \bigg\}\lan \Psi_P, \e^{-tH(P+k)} \e^{tH(P)}  \Psi_P\ran.
\end{align}
We write $\hP:=P/|P|$ and move on to polar coordinates in $P$ and $P'$ integrations:
\begin{align}
{G}^{(1)}_{T}(k,t)
&=  \sum_{\ell}\int \mrm{d}\Om(\hat{P}) \int_{\De_{\ell}}\mrm{d}|P|\, |P|^{d-1} \,\lan \Psi_{|P|\hP}, \e^{-tH(|P|\hP+k)} \e^{tH(|P|\hP)}  \Psi_{|P|\hP}\ran\non\\
&\phantom{44}\times \bigg\{  \fr{| \lan \Phi_{|P|\hP },\Psi_{|P|\hP}\ran|^2   \e^{-(2T+t)E_{\mrm{r}}(|P|) } }
{    \sum_{\ell}\int \mrm{d}\Om(\hat{P}') \int_{\De_{\ell}}\mrm{d}|P'|\, |P'|^{d-1} | \lan \Phi_{|P'| \hP'},\Psi_{|P'|\hP'}\ran|^2  \e^{-(2T+t)E_{\mrm{r}}( |P'|)} } \bigg\},
 \label{numerator-and-denominator}
\end{align}
where we also used that $P\mapsto E(P)$ is rotation invariant. 

Let us first consider a possible global minimum at zero. Since $E$ is analytic near zero, we have that  $E_{\mrm{r}}(|P|)\sim |P|^{n_0}$ {\rc  in this region} for some $n_0\in \nat_0$, $n_0\geq 2$. Thus 
an elementary analysis gives  for the numerator  in (\ref{numerator-and-denominator}) 
\beqa
& & \lim_{T\to\infty}(2T+t)^{d/n_0}\int \mrm{d}\Om(\hat{P}) \int_{\De_0}\mrm{d}|P|\, |P|^{d-1} | \lan \Phi_{|P| \hP},\Psi_{|P|\hP}\ran|^2  \e^{-(2T+t)E_{\mrm{r}}(|P|)} \times \non\\
& &\ph{444444444444444}\times \lan \Psi_{|P|\hP }, \e^{-tH(|P|\hP+k)} \e^{tH(|P| \hP)}  \Psi_{|P|\hP}\ran\non\\
& &\ph{444444444444} =C_{0}  | \lan \Phi_{0},\Psi_{0}\ran|^2  \lan \Psi_0, \e^{-tH(k)}   \Psi_0\ran  \int_{0}^{\infty} \mrm{d}{U} \e^{-U}\, U^{(d/n_0)-1}, \label{numerator-zero}
\eeqa
for some $C_{0}\neq 0$. An analogous formula holds for the denominator in (\ref{numerator-and-denominator})
\beqa
& & \lim_{T\to\infty}(2T+t)^{d/n_0}\int \mrm{d}\Om(\hat{P}) \int_{\De_0}\mrm{d}|P|\, |P|^{d-1} | \lan \Phi_{|P| \hP},\Psi_{|P|\hP}\ran|^2  \e^{-(2T+t)E_{\mrm{r}}(|P|) } \times \non\\
& &\ph{444444444444} =C_{0}  | \lan \Phi_{0},\Psi_{0}\ran|^2    \int_{0}^{\infty} \mrm{d}{U} \e^{-U}\, U^{(d/n_0)-1}. \,\,\label{denominator-zero}
\eeqa

Let us now analyse a global  minimum at $Q_{\ell}\neq 0$.  By analyticity, we have that  $E_{\mrm{r}}( {\rc Q_{\ell}+R})\sim R^{n_{\ell}}$ near $R=0$ for some $n_{\ell}\in \nat_0$,
$n_{\ell}\geq 2$. In this case we obtain for  the numerator in (\ref{numerator-and-denominator})
\beqa
& & \lim_{T\to\infty}(2T+t)^{(1/n_{\ell})}\int \mrm{d}\Om(\hat{P}) \int_{\De_{\ell}} \mrm{d}|P|\, |P|^{d-1} | \lan \Phi_{|P| \hP},\Psi_{|P|\hP}\ran|^2  \e^{-(2T+t)
E_{\mrm{r}} (|P|) }\times \non\\
& &\ph{444444444444}\times\lan \Psi_{|P|\hP }, \e^{-tH(|P|\hP+k)} \e^{tH(|P| \hP)}  \Psi_{|P|\hP}\ran\non\\
& & \ph{444}=C_{\ell} \int \mrm{d}\Om(\hat{P}) | \lan \Phi_{Q_{\ell}\hP},\Psi_{Q_{\ell}\hP}\ran|^2 \lan \Psi_{Q_{\ell}\hP}, \e^{-tH({Q_{\ell}\hP}+k)} \Psi_{Q_{\ell}\hP}\ran  \int_{0}^{\infty} \mrm{d}{U} \,\e^{-U}\, U^{ (1/n)-1 }, \quad\quad \label{numerator-one}
\eeqa
 for some $C_{\ell}\neq 0$. For the denominator in  (\ref{numerator-and-denominator}) we get in this case
 \beqa
& & \lim_{T\to\infty}(2T+t)^{(1/n_{\ell})}\int \mrm{d}\Om(\hat{P}) \int_{\De_{\ell}}\mrm{d}|P|\, |P|^{d-1} | \lan \Phi_{|P| \hP},\Psi_{|P|\hP}\ran|^2  
\e^{-(2T+t)E_{\mrm{r}}(|P|)}\times \non\\
& & \ph{4444}=C_{\ell} \int \mrm{d}\Om(\hat{P}) | \lan \Phi_{Q_{\ell}\hP},\Psi_{Q_{\ell}\hP}\ran|^2   \int_{0}^{\infty} \mrm{d}{U} \,\e^{-U}\, U^{ (1/n)-1 }. \label{denominator-one}
\eeqa
By substituting (\ref{numerator-zero})--(\ref{denominator-one}) back to formula~(\ref{numerator-and-denominator})  and considering the different cases from the
statement of the proposition, we complete the proof. \end{proof}
%%%%%%%%%%%%%%%%%%%%%%%%%%%%%%%%%%%%%%%%%%%%
\begin{lemma}\label{epsilon-limit} Suppose that $Q\geq 0$ is a global minimum of $|P|\mapsto E_{\mrm{r}}(|P|)$. Then the following relations hold:
\beqa
& &\lim_{\eps\to 0}\lan \Psi_0, \e^{-\fr{t}{\eps^2}H(\eps k)}   \Psi_0\ran= \e^{-\fr{tk^2}{2}(\pa^2_{|P|}E_{\mrm{r}})(0)} \quad\mathrm{for}\quad Q=0,  \\
& &\lim_{\eps\to 0} \lan \Psi_{Q\hP}, \e^{ -\fr{t}{\eps^2}H({Q \hP}+\eps k)} \Psi_{Q\hP}\ran= \e^{ -\fr{t}{2} (\pa^2_{|P|}E_{\mrm{r}})(Q) (\hat{P}\cdot k)^2} \quad \mathrm{for}\quad Q>0.
\eeqa
\end{lemma}
%%%%%%%%%%%%%%%%%%%%%%%%%
\begin{proof} Suppose that $Q=0$. By shifting the vector $\Psi_{0}=\Psi_{\eps k}+(\Psi_{0}-\Psi_{\eps k})$, we obtain
\begin{align}
\lan \Psi_0, \e^{-\fr{t}{\eps^2}H(\eps k)}   \Psi_0\ran&=\lan \Psi_{\eps k}, \e^{-\fr{t}{\eps^2}H(\eps k)}   \Psi_{\eps k}\ran+R(\eps)\non\\
&=\e^{-\fr{t}{\eps^2}E (\eps k)} +R(\eps) \underset{\rc \eps\to 0}{\to}  \e^{-\h tk^2 (\pa^2_{|P|} E_{\mrm{r}})(0)  },
\end{align}
where $R(\eps)$ is an error term which tends to zero as $\eps\to 0$ {\rc by assumption 3}.

Concerning the case $Q>0$, we shift the vector as follows $\Psi_{  Q\hP }=\Psi_{Q\hP+\eps k}+(\Psi_{Q\hP}-\Psi_{Q\hP+\eps k})$. This gives
\begin{align}
\lan \Psi_{Q\hP}, \e^{-\fr{t}{\eps^2} H({Q\hP}+\eps k)}  \Psi_{Q\hP}\ran&= \lan \Psi_{Q\hP+\eps k}, \e^{ -\fr{t}{\eps^2} H( {Q\hP}+\eps k)}  \Psi_{Q\hP+\eps k}\ran+R(\eps)\non\\
&= \e^{ -\fr{t}{\eps^2} E(Q\hP+\eps k) } +R(\eps) \underset{\rc \eps\to 0}{\to}   \e^{ -\fr{t}{2} (\pa^2_{|P|}E_{\mrm{r}}(Q) )(\hat{P}\cdot k)^2},
\end{align}
which completes the proof. \end{proof}
%%%%%%%%%%%%%%%%%%%%%%%%%%%%%%%%%%%%%%

\textbf{Proof of Theorem~\ref{main-general-theorem}:} We start from the case $d\geq 2$. Suppose, by contradiction, that there is no global minimum at zero. Then, from the last
part of Proposition~\ref{main-proposition} and Lemma~\ref{epsilon-limit} we obtain
\begin{align}
\e^{-\fr{tk^2 \si^2  }{2} }= \fr{ \sum_{\ell}C_{\ell} \int \mrm{d}\Om(\hat{P}) | \lan \Phi_{Q_{\ell}\hP},\Psi_{Q_{\ell}\hP}\ran|^2 
 \e^{ -\fr{t k^2}{2} (\pa^2_{|P|}E_{\mrm{r}}(Q_{\ell}) )(\hat{P}\cdot \hat{k})^2  } }
{\sum_{\ell}C_{\ell} \int \mrm{d}\Om(\hat{P}) | \lan \Phi_{Q_{\ell}\hP},\Psi_{Q_{\ell}\hP}\ran|^2  }. \label{first-equation-proof-theorem}
\end{align}
We denote $x^2:=tk^2\si^2/2$, $f_{\ell}(\hP):=C_{\ell}| \lan \Phi_{Q_{\ell}\hP},\Psi_{Q_{\ell}\hP}\ran|^2$, $m^{-1}_{\ell}:=\pa^2_{|P|}E(Q_{\ell})$. 
This gives
\begin{align} 
\sum_{\ell} \int \mrm{d}\Om(\hat{P})  f_{\ell}(\hP)&=\sum_{\ell}\int \mrm{d}\Om(\hat{P}) \, f_{\ell}(\hat P) e^{ {x^2}(1-\fr{m_{\ell}^{-1}}{\si^2}  (\hat P \cdot \hat k)^2 ) }. \label{averaging}
\end{align}
By averaging both sides w.r.t. the group of rotations we can assume that the  functions $f_{\ell}$ are constant and non-zero.
Suppose first that all $m_{\ell}^{-1}$ are zero. Then we immediately obtain a contradiction by taking $x^2\to\infty$. Now suppose that some\footnote{We note as an aside, that if some $m_{\ell}^{-1}>0$ then all $m_{\ell}^{-1}>0$ by definition of $n$.  } 
$m_{\ell_1}^{-1}>0$.
Then we obtain from (\ref{averaging})
\begin{align}
\sum_{\ell} \int \mrm{d}\Om(\hat{P})  f_{\ell}(\hP)  \geq  \int \mrm{d}\Om(\hat{P}) \, f_{\ell_1}(\hat P) \e^{ {x^2}(1-\fr{m_{\ell_1}^{-1} }{\si^2 }  (\hat P \cdot \hat k)^2 ) }
 \chi\big(1-\fr{m_{\ell_1}^{-1} }{\si^2}  (\hat P \cdot \hat k)^2 >0  \big). \label{minimum-at-zero-last-step}
\end{align}
As before, we obtain a contradiction by taking $x^2\to \infty$, due to the fact that  the functions $f_{\ell}$ are constant and non-zero.

Next, we prove part (b). Suppose, by contradiction, that there are several global minima in addition to the global minimum at zero.  
Let us assume first that case (a) from  Proposition~\ref{main-proposition} occurs, that is $n>\fr{n_{0}}{d}$.
With the help of Lemma~\ref{epsilon-limit}, we obtain from (\ref{a})
\begin{align}
\e^{-\fr{tk^2 \si^2  }{2} }= \fr{ \sum_{\ell}C_{\ell} \int \mrm{d}\Om(\hat{P}) | \lan \Phi_{Q_{\ell}\hP},\Psi_{Q_{\ell}\hP}\ran|^2 
 \e^{ -\fr{t k^2}{2} (\pa^2_{|P|}E_{\mrm{r}}(Q_{\ell}) )(\hat{P}\cdot \hat{k})^2  } }
{\sum_{\ell}C_{\ell} \int \mrm{d}\Om(\hat{P}) | \lan \Phi_{Q_{\ell}\hP},\Psi_{Q_{\ell}\hP}\ran|^2  }.
\end{align}
We obtain a contradiction by repeating the steps (\ref{first-equation-proof-theorem})--(\ref{minimum-at-zero-last-step}) above.

Now let us assume that case (b) of Proposition~\ref{main-proposition} occurs, that is $n=\fr{n_{0}}{d}$. Since $d\geq 2$ and $n\geq 2$ we obtain 
that $n_0>2$ which implies $(\pa^2_{|P|} E_{\mrm{r}})(0)=0$.   With the help of Lemma~\ref{epsilon-limit}, we obtain from (\ref{b}) using the
notation introduced above
\beqa
& & c_{0}  | \lan \Phi_{0},\Psi_{0}\ran|^2+  \sum_{\ell} \int \mrm{d}\Om(\hat{P}) f_{\ell}(\hat P)  \non\\
& &= c_{0}  | \lan \Phi_{0},\Psi_{0}\ran|^2 \e^{x^2}+ \sum_{\ell} \int \mrm{d}\Om(\hat{P}) f_{\ell}(\hat P)
 \e^{ {x^2}(1-\fr{m_{\ell}^{-1} }{\si^2 }  (\hat P \cdot \hat k)^2 ) }.     \label{case-b}
\eeqa
Due to the presence of the non-zero terms involving $c_{0}  | \lan \Phi_{0},\Psi_{0}\ran|^2$, we immediately obtain a contradiction by taking $x^2\to\infty$.

Finally, suppose that we are in case (c)  of Proposition~\ref{main-proposition}, that is $n<\fr{n_{0}}{d}$. Also in this case we have $n_0>2$ which implies $(\pa^2_{|P|} E_{\mrm{r}})(0)=0$. By equation (\ref{c}) and  Lemma~\ref{epsilon-limit},  we obtain $\e^{-x^2}=1$ which is immediately a contradiction.
This concludes the proof of part (b) of the theorem.

Given that we have only one global minimum, we can apply formula (\ref{one-minimum-formula}). Together with Lemma~\ref{epsilon-limit}, we obtain
\beqa
\e^{-\fr{tk^2 \si^2  }{2} }= \e^{-\fr{tk^2 (\pa^2_{|P|} E_{\mrm{r}})(0)     }{2} },
\eeqa
which gives part (a) of the theorem.

For $d=1$, the reasoning above requires several modifications. From the assumption that there is no global minimum at zero
we obtain via (\ref{averaging}) that $m^{-1}_{\ell}=\si^2$ for $\ell=0,1,\ldots, L$, which is what we wanted to prove. 

Now suppose that there is a global minimum at zero and possibly some non-zero global minima. 
In the case $n>\fr{n_{0}}{d}$ from the fact that $n_0\geq 2$
we conclude that $n>2$, hence $m^{-1}_{\ell}=0$ for $\ell\neq 0$.   In this situation formula (\ref{averaging}) gives directly a
contradiction. 

In the case $n=\fr{n_{0}}{d}$ we distinguish two sub-cases. First, for $n=n_0>2$  we have $m_{\ell}^{-1}=0$ for
all $\ell$, including $\ell=0$, and thus formula (\ref{case-b}) gives a contradiction. Second, for $n=n_0=2$ we have $m_{\ell}^{-1}\neq 0$
and thus formula~(\ref{case-b}) has to be rewritten as follows
\beqa
& & c_{0}  | \lan \Phi_{0},\Psi_{0}\ran|^2+  \sum_{\ell} \int \mrm{d}\Om(\hat{P}) f_{\ell}(\hat P)  \non\\
& &= c_{0}  | \lan \Phi_{0},\Psi_{0}\ran|^2 \e^{x^2(1-\fr{m_{0}^{-1} }{\si^2 }   )   }+ \sum_{\ell} \int \mrm{d}\Om(\hat{P}) f_{\ell}(\hat P)
 \e^{ {x^2}(1-\fr{m_{\ell}^{-1} }{\si^2 }   ) },     \label{case-b-1}
\eeqa
where angular integration denotes now summation over $\hat{P}=\pm 1$. Clearly, we avoid a contradiction  iff $m_{\ell}^{-1}=\si^2$
for $\ell=0,1,\ldots, L$.  (It is important here that if $n=\mrm{max}_{\ell\neq 0}(n_{\ell})=2$ then $ n_{\ell}=2$ for all $\ell\neq 0$). 

In the case  $n<\fr{n_{0}}{d}$ we obtain a contradiction as before. Thus in the case $d=1$ the assumption that there are several
global minima of $E_{\mrm{r}}$ led us to the conclusion that the inverses of their effective masses $ \pa^2_{|P|}E_{\mrm{r}}(Q_{\ell})$, $\ell=0,1,\ldots L$,
must all be equal to $\si^2$.  \null\hfill\qed\\

%%%%%%%%%%%%%%%%%%%%%%%%%%%%%%%%%%%
\noindent \textbf{Acknowledgements.}  We thank Fumio Hiroshima, Tadahiro Miyao, Chiranjib  Mukherjee,  and Jacob Schach M\o ller
for helpful discussions. This work was partially supported by the Deutsche Forschungsgemeinschaft (DFG) within the grant
DY107/2-1.
%%%%%%%%%%%%%%%%%%%%%%%%%%%%%%%%%%%

\appendix

%%%%%%%%%%%%%%%%%%%%%%%%%%%%%%%%%%%%%%%%%
\section{Analyticity  of the Fr\"ohlich polaron  in total momentum} \label{holomorphy}
%%%%%%%%%%%%%%%%%%%%%%%%%%%%%%%%%%
\setcounter{equation}{0}

In this appendix we verify that $P\mapsto H^{\mrm{Fr}}(P)$ is a real analytic family as specified in the standing assumption 0 of Section~\ref{General-analysis}.

%%%%%%%%%%%%%%%%%%%%%%%%%%%%%%%%%%%%%%%%
\begin{lemma}\label{resolvent-set-lemma} Suppose that $\xi \notin \mrm{spec}(H(P_0))$. Then $\xi \notin  \mrm{spec}(H^{\mrm{Fr}}(P))$ for 
$P$ in a neighbourhood  $N_{P_0}$ of $P_0$. The function $N_{P_0}\ni  P\mapsto (H^{\mrm{Fr}}(P)-\xi)^{-1}$ is real analytic.
\end{lemma}
%%%%%%%%%%%%%%%%%%%%%%%%%%%%%%%%%%%%%%%%%
\begin{proof} First, suppose that $\xi\notin \mrm{spec}(H_{\ka}(P))$  for $\ka$ sufficiently large and note that on $D(P_{\mrm{f}}^2+N_{\mrm{f}})$
\beqa
H_{\ka}(P)=H_{\ka}(P_0)+ \h (P-P_0)^2-(P-P_0)\!\cdot\! (P_{\mrm{f}}-P_0).
\eeqa
Consequently, for $P$ in a small neighbourhood  $N_{P_0}$ of $P_0$, which a priori may depend on $\ka$,  the series on the r.h.s. below
converges and defines the inverse of $(H_{\ka}(P)-\xi)$:
\beqa
\fr{1}{H_{\ka}(P)-\xi}=\fr{1}{H_{\ka}(P_0) -\xi}  \sum_{n=0}^{\infty}\!\bigg\{\!\! -(\h (P-P_0)^2-(P-P_0)\!\cdot\! (P_{\mrm{f}} -P_0)  ) \fr{1}{H_{\ka}(P_0) -\xi}\bigg\}^n.\,\, \label{expansion-in-P}
\eeqa
To eliminate the dependence of $N_{P_0}$ on $\ka$, we  show in Lemma~\ref{P-f-bound} below, that  
%%%%%%%%%%%%%%%%%%%%%%%%%%%%%%%%%%%%%%
\beqa
\|P_{\mrm{f},j }\fr{1}{H_{\ka}(P_0) -\xi}\|\leq c, \quad j=1,2,3, \label{uniform-bound-P}
\eeqa
uniformly in $\ka$.

Now we intend to take the limit $\ka\to\infty$ on both sides of (\ref{expansion-in-P}). By  \cite[Theorem~VIII.23]{RSI}, if $\xi \notin \mrm{spec}(H^{\mrm{Fr}}(P) )$ then $\xi \notin\mrm{spec}(H_{\ka}{(P)})$ for 
$\ka$ sufficiently large and $ (H_{\ka}{ (P)}-\xi)^{-1}  \to (H^{\mrm{Fr}}{(P)}-\xi)^{-1}$ in norm.  As the same is true for $P$ replaced with $P_0$ and we can use (\ref{uniform-bound-P}) to exchange
the limit  $\ka\to\infty$ with summation in (\ref{expansion-in-P}).  Thus the proof is complete. \end{proof}
%%%%%%%%%%%%%%%%%%%%%%%%%%%%%%%%%%%%
\begin{lemma} \label{P-f-bound} The following bounds hold uniformly in $\ka$:
\beqa
\|P_{\mrm{f},j} (H_{\ka}(P)+\i)^{-1}\|\leq c, \quad j=1,2,3.
\eeqa
\end{lemma}
%%%%%%%%%%%%%%%%%%%%%%%%%%%%%%%%%%%%%%%%%%%%%%%%%%%%%%%
\begin{proof} First{\rc ,} we recall some  material from \cite[Appendix A]{Mi10}, referring there for more details. Let  
\beqa
T_{K,\ka}=\int_{{\rc \real^3}} \mrm{d}k\, \be_{K,\ka}[a(k)-a^*(k)], \quad \be_{K,\ka}(k)= {\rc -}\sqrt{\al}{\rc \fr{1}{\sqrt{2}\pi} } \fr{\chi_{[0,\ka]}(|k|)}{|k|(1+k^2/2)} \chi_{[K,\infty)}(|k|),
\eeqa
where  $K$ is chosen  sufficiently large (depending on $\al$ but not on $\ka$ or $P$) 
as specified above Lemma~A.3 of \cite{Mi10}. Let $H^{\mrm{free}}(P)$ denote the Hamiltonians (\ref{cut-off-Hamiltonians}) with $\al=0$. Then the Gross-transformed  Hamiltonians $\ti{H}_{\ka}{\rc (P)}:=e^{T_{K,\ka}} H_{\ka}(P) e^{-T_{K,\ka}}$ 
{\rc are self-adjoint operators  which converge in the norm-resolvent sense to a limiting Hamiltonian $\ti{H}({\rc P})$ \cite[Proposition A.4]{Mi10}. 
As stated in the proof of this latter proposition, 
   $H^{\mrm{free}}(P)\leq C'(\ti{H}_{\ka}{\rc (P)}+C)$,
with $ C, C'$ independent of $\ka$. Hence,
\beqa
\| |P-P_{\mrm{f}}|   (\ti{H}_{\ka}{\rc (P)}+C)^{-1/2}\|,  \| N_{\mrm{f}}^{1/2}  (\ti{H}_{\ka}{\rc (P)}+C)^{-1/2}\|\leq C'. \label{bounds-from-Miyao}
\eeqa
Next, we can write on  $D(H^{\mrm{free}}(P))$
\beqa
e^{T_{K,\ka}}  P_{\mrm{f}} e^{-T_{K,\ka}}=P_{\mrm{f}}+a(k\be_{K,\ka})+a^*(k\be_{K,\ka})  + \int_{\rc \real^3} \mrm{d}k\,   k\,|\be_{K,\ka}(k)|^2,
\eeqa
where the last term is actually zero by symmetry. Noting that $  \int \mrm{d}k\,   \,|k\be_{K,\ka}(k)|^2    \leq c$, uniformly in $\ka$, we have   $\| a^{(*)}(k\be_{K,\ka})(1+N_{\mrm{f}})^{-1/2}\|\leq c$ uniformly in $\ka$.
Therefore, by estimates~(\ref{bounds-from-Miyao}),
\begin{align}
\| P_{\mrm{f},j}  (H_{\ka}(P)+C)^{-1}   \|  &\leq C_0( \| P_{\mrm{f},j }   (\ti{H}_{\ka}(P)+C)^{-1}\|
+ \|(1+ N_{\mrm{f}}  )^{1/2}  (\ti{H}_{\ka}(P)+C)^{-1}  \|)
%&\leq C_1( \|      (1+H^{\mrm{free}}(P) )   (\ti{H}_{\ka}(P)+\i)^{-1}   \|^{1/2}+1)
\end{align}
is   uniformly bounded in $\ka$. This concludes the proof. }\end{proof}

%%%%%%%%%%%%%%%%%%%%%%%%%%%%%%%%%%%%%%%%%%%%%%%
\section{Proof of Proposition~\ref{corollary}} \label{proposition}
%%%%%%%%%%%%%%%%%%%%%%%%%%%%%%%%%%%%%%%%%%%%%%%
\setcounter{equation}{0}

 Let us consider first polaron-type models satisfying Condition C. For the standing assumption 0 {\rc we refer to the discussion of Lemma~\ref{Jacob}, part 6.
By part 2 of Lemma~\ref{Jacob},   $\mathcal{I}_0$ contains neighbourhoods of all the global minima of $|P|\mapsto E_{\mrm{r}}(|P|)$. Thus, considering other items
of this lemma,} it suffices to show that there
is a finite number of such minima to complete the proof of  Proposition~\ref{corollary} (a). {\rc To this end, we first note that by part 1 of Lemma~\ref{Jacob},
$P\mapsto E_{\mrm{ess}}(P)$ is bounded if $\om$ is bounded}. Now by
 part 4 of the same lemma, all global minima
must be localized in a compact set. (In particular, $P\mapsto E(P)$ cannot be a constant function in an open set, hence everywhere).
Now suppose there is a finite accumulation point $(Q_*, E_{\mrm{r}}(Q_{*}))$ of the set of global minima. Since the spectrum is closed, the corresponding sphere belongs to the spectrum. By Lemma~\ref{Jacob} parts 2 and 3, for $|P|=Q_*$, $E(P)$ is an isolated eigenvalue of $H(P)$.
Thus, by analyticity, this eigenvalue can be  continued to some neighbourhood of $Q_*$ in the radial direction.  As any neighbourhood 
contains infinitely many  global minima, we conclude that $P\mapsto E(P)$ is constant near $|P|= Q_{*}$, hence everywhere, which is a contradiction.

Let us now move on to part (b).
For the models in question we have $\Pi=F \e^{\mrm{i} P_{\mrm{f}}x } $, where $F$ is the unitary Fourier transform from $x$ to $P$ variable.
We decompose $\varphi \otimes\Omega=\Pi^{*}\int^{\oplus}_{\real^d} \mrm{d} P\phi_P$. Thus to find $\phi_P$ we compute
\beqa\label{Pi-on-state}
\Pi(\varphi \otimes\Omega)=\int^{\oplus}_{\real^d}\mrm{d}P \,  \hat{\varphi}(P)\Om,
\eeqa 
which gives $\phi_P=  \hat{\varphi}(P)\Om$. The function 
$P\mapsto |\lan  \phi_P, \psi_P\ran|= |\hat\varphi(P)|  \, |\lan  \Om, \psi_P\ran|$
is continuous by continuity of $\hat\varphi$ and analyticity of $P\mapsto |\psi_P\ran \lan \psi_P|$. It is non-zero
by our assumption on $\hat\varphi$ and by part 5 of Lemma~\ref{Jacob}.  

Concerning the Fr\"ohlich polaron, standing assumption 0 is proven in Appendix~\ref{holomorphy}. The  other claims are verified as above,
using Lemma~\ref{Froehlich-theorem}. 
%%%%%%%%%%%%%%%%%%%%%%%%%%%%%


\begin{thebibliography}{99} 

\bibitem{AD10} A.S. Alexandrov and J.T. Devreese, Advances in Polaron Physics, Springer, 2010.

\bibitem{BS05} V. Betz and H. Spohn, A central limit theorem for Gibbs measures relative to Brownian
motion. Prob. Theory. Rel. Fields. \textbf{131}, 459 -- 478 (2005).


\bibitem{BT17}  G.A. Bley and  L.E. Thomas, Estimates on functional integrals of quantum mechanics
and non-relativistic quantum field theory, Commun. Math. Phys. \textbf{350}, 79 -- 103 (2017).

\bibitem{F55} R.P. Feynman, Slow electrons in a polar crystal, Phys. Rev. \textbf{97}, 660 -- 665 (1955).

\bibitem{F} R. Feynman, Statistical Mechanics, a set of lectures. W. A. Benjamin, Inc., 1972. 

\bibitem{F74} J. Fr\"{o}hlich, Existence of dressed one electron states in a class of persistent models. Fortschritte der Physik \textbf{22}, 159 -- 198 (1974).


\bibitem{GL91} B. Gerlach and H. L\"{o}wen,
Analytical properties of polaron systems or: do polaron phase transition exist or not? Rev. Mod. Phys. \textbf{63}, 63 -- 90 (1991).

\bibitem{GW16} M. Griesemer and A. W\"unsch, Self-adjointness and domain of the Fr\"ohlich Hamiltonian. J. Math. Phys. \textbf{57}, 021902 (2016).

\bibitem{Gro72} E.P. Gross,  Existence and uniqueness of physical ground states. J. Funct. Anal. \textbf{10}, 52--109 (1972).

\bibitem{G06} M. Gubinelli, Gibbs measures on self-interacting Wiener paths. Markov Processes and Related Fields \textbf{12}, 747 -- 766  (2006).

\bibitem{Ka} T. Kato, Perturbation Theory for Linear Operators. Springer, 1995.

\bibitem{LS19} E.H. Lieb and R. Seiringer, Divergence of the effective mass of a polaron in the strong coupling limit.
\texttt{arXiv}:1902.04025.


\bibitem{Mi10} T. Miyao, Nondegeneracy of ground states in nonrelativistic quantum field theory. J. Operator Theory \textbf{64}, 207 -- 241 (2010). 

\bibitem{SM14} J.S. M{\o}ller, The polaron revisited, Rev. Math. Phys. \textbf{18}, 485 - 517 (2006).

\bibitem{MR12} J.~S. M{\o}ller and M.~G. Rasmussen, The translation invariant massive Nelson model: II. The continuous spectrum below the two-boson threshold,  Ann. Henri Poincar\'e, \textbf{14} 793 -- 852 (2013). 

\bibitem{MV18} C. Mukherjee and S.R.S. Varadhan, Identification of the polaron measure and its central limit theorem. \texttt{arXiv}:1802.05696.

\bibitem{MV18a} C. Mukherjee and S.R.S. Varadhan, Identification of the polaron measure in strong coupling and the Pekar variational formula.
\texttt{arXiv}:1812.06927.
 
 \bibitem{M19} C. Mukherjee,
 A central limit theorem for Gibbs measures including long range and singular interactions and homogenization of the stochastic heat equation.
\texttt{arXiv}: 1706.09345v3. 

\bibitem{P54} S.I. Pekar, Untersuchungen zur Elektronentheorie der Kristalle. Akademie Verlag, Berlin, 1954.

\bibitem{RSI} M. Reed and B. Simon, Methods of Modern Mathematical Physics I. Functional Analysis. Academic Press, 1972.

\bibitem{Ro94} G. Roeppsdorff, Path Integral Approach to Quantum Physics. Springer-Verlag, 1994.

\bibitem{S87} H. Spohn, Effective mass of the polaron: a functional integral approach. Ann. Phys. \textbf{175}, 278 -- 318 (1987).

\bibitem{S88} H. Spohn, The polaron at large total momentum. J. Phys. A: Math. Gen. \textbf{21} 1199 -- 1211 (1988).

\bibitem{Sp} H. Spohn, Dynamics of Charged Particles and Their Radiation Field. Cambridge University Press, 2004.
 
\end{thebibliography}
\end{document}